\documentclass[reqno]{amsproc}

\usepackage{amstext,amsmath,amssymb,amsfonts}
\usepackage[latin1]{inputenc}
\usepackage{epsfig}
\usepackage{hyperref}
\usepackage{color}

\usepackage{latexsym}

\newtheorem{lemma}{Lemma}

\newtheorem{theorem}{Theorem}

\newcommand{\bea}{\begin{eqnarray}}
\newcommand{\eea}{\end{eqnarray}}
\newcommand{\beq}{\begin{equation}}
\newcommand{\eeq}{\end{equation}}
\newcommand{\enn}{\nonumber \end{equation}}

\title[Generalizing Bagarello's operator for PDEs]
{Generalizing Bagarello's operator approach to solve a class of  partial 
differential equations\footnote{Preprint: ICMPA-MPA/2015/10 }}

\author{Jean Ghislain Compaor\'e}
\address[J.G.C.]{International Chair in Mathematical Physics and Applications,
ICMPA-UNESCO Chair, 072 BP 50, Cotonou, Rep. of Benin}
\email{ghislaincompaore@yahoo.fr}

\author{Villevo Adanhounm\`e}
\address[V.A.]{International Chair in Mathematical Physics and Applications,
ICMPA-UNESCO Chair, 072 BP 50, Cotonou, Rep. of Benin}
\email{adanhounmvillvo@yahoo.fr}%

\author{Mahouton  Norbert Hounkonnou}
\address[M.N.H.]{International Chair in Mathematical Physics and Applications,
ICMPA-UNESCO Chair, 072 BP 50, Cotonou, Rep. of Benin}
\email{norbert.hounkonnou@cipma.uac.bj}

\begin{document}

\maketitle 

\begin{abstract}
The non-commutative strategy developed by Bagarello
(see {\it Int. Jour. of Theoretical Physics}, \textbf{43}, issue 12 (2004), 
p. 2371 - 2394) for the analysis of systems of ordinary differential 
equations (ODEs) is extended to a class of 
partial differential equations (PDEs), namely evolution equations 
and Navier-Stokes equations. Systems of 
PDEs are solved using an unbounded self-adjoint, densely defined, Hamiltonian operator and  a recursion relation which provides  a multiple 
commutator and  a power series solution.
Numerous examples are given in this work.
\end{abstract}

\tableofcontents

\section{Introduction}
Evolution equations and linearized Navier-Stokes equations are extremely important in Mathematics and Physics, especially for applications in engineering as they are often used to obtain vital information on flow phenomena such as the velocity distribution, the flow pattern and pressure losses. However, it is a well known fact that exact solutions of evolution equations are rather difficult to obtain. Equally, till date, very few exact solutions of the Navier-Stokes equations are known. This is largely due to the complexity of the system of differential equations involved.

In the absence of a general solution, it is often convenient to experiment with models so as to obtain information on the flow phenomena. Chengri Jin and Mingzhu Liu \cite{Chengri} used the Adomian decomposition method to solve a kind of evolution equation. Lu \cite{Lu} applied the \textquotedblleft background\textquotedblright method of the arbitrary Prandtl number problem to the Navier-Stokes equations in order to derive a scaling lower bound on the space-time averaged temperature of the layer along with an explicit prefactor and then used a multiple boundary layer asymptotic theory to sharpen the estimate, thus increasing the prefactor in the lower bound by another factor. The time-dependent discretized Navier-Stokes equations with Chorin's projection scheme using the continuous Galerkin finite element method in time and space were studied by Kharrat and Mghazli\cite{Kharrat}. 
More recently, the existence of global regular solutions to Navier-Stokes equations was proven by Zadrzynska and Zajaczkowski \cite{Zadrzynska} while the exact solutions of rotating shallow water were obtained by Chesnokov \cite{Chesnokov} using group analysis.

In this paper, we investigate the analytical solutions to two categories of PDE using Bagarello's approach. The first category comprises evolution equations which arise in many modeling problems in physics and hydrodynamics. Well known examples include the KdV, MKdV, BBM, Burgers, regularized long wave (RLW) and KdV-KSV equations, which are all special cases of the general equation in \cite{Chengri}. The second category is made up of Navier-Stokes equations and heat propagation equations. Equations in this category are used to describe phenomena in continuum mechanics and are derived from the mathematical model of the diffusion of particles in a medium.
 
The rest of the paper is organized as follows. In Section 2, we present some details of the models we analyze. In Section 3, we obtain the analytical solutions to evolution equations, linearized Navier-Stokes equations and the pressure using a non-commutative approach. The validity of the technique is verified through illustrative examples. This is followed by the conclusion.
 
\section{Analytical solution}
Our aim in this section is to solve all the types of equations enumerated above.
\subsection{Evolution equations}
We consider the general evolution equation of the form
\begin{eqnarray} 
\dfrac{\partial u}{\partial t}=\sum_{m= 0}^Ma_m\dfrac{\partial^m u}{\partial x^m}
+\sum_{m= 0}^N b_m\dfrac{\partial^m u^{k+1}}{\partial x^m}
+ c\dfrac{\partial^{i+1} u}{\partial x^i\partial t} +f(x,t),\,\,-\infty < x < +\infty,\,\, t> 0 \label{1} 
\end{eqnarray}
subjected to the initial condition
\begin{eqnarray} u(x,t)|_{t=0}= h(x)\label{2} \end{eqnarray}
where $f$ and $h$ are given functions, $a_m, b_m$ and $c$ are real constants, 
$k$ and $i$ are positive integers, $M$ and $N$ are nonnegative integers.

We suppose that the solution of the previous Cauchy problem exists and is 
unique. Then we solve this problem by Bagarello's approach \cite{Bagarello97}, 
based on a suggestion drawn from quantum mechanics. Indeed, given a 
quantum mechanical system $S$ and the related set of observables 
$O_S$, that is, the set of all the self-adjoint bounded (or more often unbounded) 
operators related to $S$, the evolution of any observable $Y\in O_S$ 
satisfies the Heisenberg equation of motion (HOEM)
\begin{eqnarray}
\dfrac{\partial}{\partial t} Y (t, x) = i[H,Y(t, x)]. \end{eqnarray}
Here $[A,B] := AB- BA$ is the commutator between $A, B\in O_S$, $H$ is 
assumed to be a densely defined self-adjoint Hamiltonian operator of 
the system acting on some Hilbert space $\mathcal{H}.$ The initial 
condition $Y^0=h(x)$ is also considered as an operator acting on the 
Hilbert space $\mathcal{H}$.\\
The following statement holds:
  \begin{theorem} A formal solution of the HOEM is 
$ Y(t,x)=e^{itH}Y^0e^{-itH},$ where $Y^0$ is the initial condition of 
$Y(t,x)$ and $H$ does not depend explicitly on time.
 Furthermore if $H$ is bounded, we get $Y (t,x) =\sum_{k=0}^{\infty}
\dfrac{ [it)^k}{k!}[H,Y^0]_k$ where $[A,B]_k$ is the multiple commutator 
defined recursively as: $[A,B]_0 = B; [A,B]_k = [A,[A,B]_{k-1}]$.
\end{theorem} 
\begin{proof} Indeed, we have  \begin{eqnarray*}
\dfrac{\partial}{\partial t} Y (t,x)&=&\dfrac{\partial}{\partial t}[e^{itH}]Y^0e^{-itH}+
e^{itH}Y^0\dfrac{\partial}{\partial t}[e^{-itH}]\Longleftrightarrow\\
\dfrac{\partial}{\partial t} Y (t,x)&=& iHe^{itH}Y^0e^{-itH}+
e^{itH}Y^0(-iH)e^{-itH}\Longleftrightarrow\\
\dfrac{\partial }{\partial t} Y (t,x)&=& i[H, Y (t,x)]. \end{eqnarray*}
If $H$ is bounded, $Y(t,x)=e^{itH}Y^0e^{-itH}$ can be expanded as
\begin{eqnarray*} Y (t,x) =\sum_{k=0}^{\infty} \dfrac{(it)^k}{k!}[H,Y^0]_k.
\end{eqnarray*}\end{proof}
In order to apply Bagarello's technique to the initial value problem 
we need the following  
\begin{lemma} Assume $$ u(x,t)=\sum_{n\geq 0}\dfrac{(it)^n}{n!}w_n^1,\,\,\,
 w_n^1=[H, Y^0]_n.$$
Then, we have
\begin{eqnarray*} u^p(x,t)= \sum_{n\geq 0}\dfrac{(it)^n}{n!}w_n^p,
\end{eqnarray*}
where $ w_n^p= \sum_{k=0}^n C_n^kw_k^1w_{n-k}^{p-1},\,\,\, p\geq 2, 
\,\,\, C_n^k=\dfrac{n!}{k!(n-k)!}$.\end{lemma}
\begin{proof}
Using the product of series
\begin{eqnarray*} \sum_{j\geq 0}\sum_{k\geq 0}\dfrac{B^j}{j!}\dfrac{A^k}{k!}=
\sum_{n\geq 0}\sum_{k=0}^n\dfrac{1}{k!(n-k)!}B^kA^{n-k},\end{eqnarray*}
we can write
\begin{eqnarray*}
u^2(x,t)&=& \sum_{j\geq 0}\sum_{k\geq 0}\bigg\{\dfrac{(it)^j}{j!}[H, Y^0]_j\bigg\}
\bigg\{\dfrac{(it)^k}{k!}[H, Y^0]_k\bigg\}\\
&=& \sum_{n\geq 0}\sum_{k=0}^n\dfrac{1}{k!(n-k)!}(it)^k[H, Y^0]_k(it)^{n-k}
[H, Y^0]_{n-k}= \sum_{n\geq 0}\dfrac{(it)^n}{n!}w_n^2;\\
u^3(x,t)&=& \sum_{j\geq 0}\sum_{k\geq 0}\bigg\{\dfrac{(it)^j}{j!}[H, Y^0]_j\bigg\}
\bigg\{\dfrac{(it)^k}{k!}w_k^2\bigg\}\\
&=& \sum_{n\geq 0}\sum_{k=0}^n\dfrac{1}{k!(n-k)!}(it)^k[H, Y^0]_k(it)^{n-k}w_{n-k}^2
= \sum_{n\geq 0}\dfrac{(it)^n}{n!}w_n^3,
\end{eqnarray*}
where $ w_n^2= \sum_{k=0}^n C_n^k w_k^1w_{n-k}^1$ and $ w_n^3= 
\sum_{k=0}^n C_n^kw_k^1w_{n-k}^2$.\\ By induction we get 
\begin{eqnarray*} u^p(x,t)= \sum_{n\geq 0}\dfrac{(it)^n}{n!} w_n^p,\,\,\, w_n^p
= \sum_{k=0}^n C_n^kw_k^1w_{n-k}^{p-1},\,\,\, p\geq 2.\end{eqnarray*}
In fact \begin{eqnarray*}  u^{p+1}(x,t)&=& u(x,t)u^p(x,t)=
\sum_{n\geq 0}\sum_{k=0}^n \dfrac{1}{k!(n-k)!}(it)^kw_k^1 (it)^{n-k}w_{n-k}^p,
\end{eqnarray*} yielding
\begin{eqnarray*} u^{p+1}(x,t)&=& \sum_{n\geq 0}\dfrac{(it)^n}{n!}
\sum_{k=0}^nC_n^kw_k^1 w_{n-k}^{(p+1)-1} \end{eqnarray*}
\end{proof}
Now we state the first of the main results through the following:
\begin{theorem}
If $u(x, t)=\sum_{n\geq 0}\dfrac{(it)^n}{n!}w_n^1$ is a solution of 
the problem (\ref{1}), (\ref{2}), for $f(x, t)=0$, then $w_n^1$ 
must satisfy the following recursive relation: 
\begin{eqnarray}  w_0^1&=& Y^0\\ 
 i w_{n+1}^1- ic \dfrac{\partial^i}{\partial x^i}w_{n+1}^1&=&
\sum_{m= 0}^Ma_m\dfrac{\partial^m }{\partial x^m}w_n^1
+\sum_{m=0}^N b_m\dfrac{\partial^m }{\partial x^m}w_n^{k+1},\,\, \forall n\geq 0,\label{3}
\end{eqnarray}
where $ w_n^1= [H, Y^0]_n,\,\,\,w_n^p= \sum_{k=0}^n C_n^kw_k^1w_{n-k}^{p-1}.$
\end{theorem}
\begin{proof} Setting the expansions of $u(x, t)$ and $u^{k+1}(x, t)$ 
in the equation (\ref{1}),  we get 
  \begin{eqnarray*}
\sum_{n\geq 0}\dfrac{(it)^n}{n!}\Bigg\{i w_{n+1}^1- ic \dfrac{\partial^i}{\partial x^i}
w_{n+1}^1-\sum_{m= 0}^Ma_m\dfrac{\partial^m }{\partial x^m}w_n^1
-\sum_{m=0}^N b_m\dfrac{\partial^m }{\partial x^m}w_n^{k+1}\Bigg\}=0,
\end{eqnarray*}
which allows the following recursive relation 
 \begin{eqnarray*}  w_0^1&=& Y^0\\
i w_{n+1}^1- ic \dfrac{\partial^i}{\partial x^i}w_{n+1}^1&=&\sum_{m= 0}^Ma_m
 \dfrac{\partial^m }{\partial x^m}w_n^1 +\sum_{m=0}^N b_m\dfrac{\partial^m }{\partial x^m}w_n^{k+1}
 \end{eqnarray*} \end{proof}
\paragraph{Remark} One can find $ w_n^1,\,\,\,n=1, 2, 3,\ldots $ as 
particular solutions of the relation (\ref{3}). 
\subsection{Diffusion equations}
Processes of heat propagation and of diffusion of particles in a medium are 
described by the general diffusion equation \cite{Agoshkov}: \begin{eqnarray}
\rho\dfrac{\partial u}{\partial t}-\nabla.(p \nabla u)+ q u = F(x,y,z,t),\label{4}
\end{eqnarray}
where $\rho$ is the coefficient of porosity of the medium and $p$ and 
$q$ characterize its properties.
In examining heat propagation, by $u(x,y,z,t)$ we denote the temperature 
of the medium at the point $(x, y, z)$ at time $t$. 
Assuming that the medium is isotropic, $\rho(x,y,z), c(x,y,z)$ and $k(x,y,z)$  
denote  its density, specific heat capacity and heat conduction 
coefficient, respectively,  while $F(x,y,z,t)$ represents the intensity of heat sources. 
The process of heat propagation is described by a function which satisfies an equation of the type:
\begin{eqnarray} c\rho\dfrac{\partial u}{\partial t}-\nabla.(k \nabla u) 
= F(x,y,z,t).\label{5} \end{eqnarray}
If the medium is homogeneous, i.e. $c, \rho$ and $k$ are constant, 
equation (\ref{5}) has the form:
\begin{eqnarray} \dfrac{\partial }{\partial t}u= a^2\Delta u + f(x,y,z,t),\label{6}
\end{eqnarray}
where $a^2=\dfrac{k}{c\rho}$, $f=\dfrac{F}{c\rho}$ and the last equation 
is the heat conductivity equation. For a complete description of the sources 
of heat propagation, it is necessary to specify the initial distribution 
of temperature $u$ in the medium (initial conditions) and the conditions 
at the boundary of this medium (boundary conditions).\\
Suppose that the equation (\ref{6}) is subjected to the initial condition 
\begin{eqnarray}
u|_{t=0}=u_0(x,y,z).\label{7}
\end{eqnarray}
For given functions $u_0$ and $ f,$ one can solve the initial value problem using Bagarello's approach. 

Here we formulate the second result:
\begin{theorem}
Assume that $u_0\in\mathbf{C}^{\infty}(\mathbf{R}^3)$, the series $$\sum_{k=0}^{+\infty}\dfrac{\delta^k}{k!}\Delta^ku_0(x, y, z),\,\,\, \delta > 0$$ and those obtained by differentiating termwise up to second order converges uniformly on each bounded and closed subset $\Omega\subset\mathbf{R}^3$.
 Let $$ u_1(x,y,z,t)=\sum_{k=0}^{+\infty}\dfrac{(a^2)^kt^k}{k!}\Delta^ku_0(x,y,z),
 \,\,\,0< t< \dfrac{\delta}{a^2}$$ be a solution of the equation
\begin{eqnarray}  \dfrac{\partial u}{\partial t}= a^2\Delta u \label{8} 
\end{eqnarray} 
satisfying the initial condition (\ref{7}) where $\Delta$ is the Laplacian operator.\\
If $$u_2(x,y,z,t)=\sum_{k\geq 0}\dfrac{(it)^k}{k!}w_k $$ is another 
solution of equation $(\ref{8})$ satisfying the condition (\ref{7}), then we have  
\begin{eqnarray}
w_k =(-ia^2)^k\Delta^k u_0(x,y,z ),\,\,\,k\geq 0, \label{9}
\end{eqnarray}
where $w_k=[H,Y^0]_k$ and $Y^0= u_0(x,y,z ).$ 
\end{theorem}
\begin{proof}
Indeed,
\begin{eqnarray*}
\dfrac{\partial u_1}{\partial t}(x,y,z,t)&=& a^2 \sum_{k=0}^{+\infty}
\dfrac{(a^2)^kt^k}{k!}\Delta^{k+1}u_0(x,y,z)\Longleftrightarrow\\
\dfrac{\partial u_1}{\partial t}(x,y,z,t)&=&a^2\Delta\bigg[\sum_{k=0}^{+\infty}
\dfrac{(a^2)^kt^k}{k!}\Delta^ku_0(x,y,z)\bigg]\Longleftrightarrow\\
\dfrac{\partial u_1}{\partial t}(x,y,z,t)&=& a^2\Delta u_1(x,y,z, t).
\end{eqnarray*}
 Then, $u_1$ is a solution of (\ref{8}) such that $u_1(x,y,z, t)|_{t=0}=u_0(x,y,z)$ .
 On the other hand, if $u_2$ is a solution of (\ref{8}) satisfying $u_2(x,y,z, t)|_{t=0}=u_0(x,y,z),$  we can write $$
  \dfrac{\partial u_2}{\partial t}(x,y,z,t)= \sum_{k=0}^{+\infty}\dfrac{t^k}{k!}i^{k+1}w_{k+1},\,\,\, a^2\Delta u_2= \sum_{k=0}^{+\infty}\dfrac{t^k}{k!}i^ka^2\Delta w_k$$
which implies 
\begin{eqnarray}
 w_{k+1}= -ia^2\Delta w_k.  
\end{eqnarray}
The last relation allows to write  
\begin{eqnarray*}
w_1= -ia^2\Delta Y^0,\,\,\, w_2=(-ia^2)^2\Delta^2 Y^0,\,\,\,
  w_3= (-ia^2)^3\Delta^3 Y^0. 
\end{eqnarray*}
 By induction, we have $w_k= (-ia^2)^k\Delta^kY^0 $.
  By virtue of the uniqueness of the solution of Cauchy's problem, we obtain
$u_1(x,y,z,t)=u_2(x,y,z,t)$ which implies 
\begin{eqnarray*}
 i^kw_k = (a^2)^k\Delta^ku_0(x,y,z).
  \end{eqnarray*}
 Thus, $w_k$ can be expressed as 
\begin{eqnarray*}
 w_k = (-ia^2)^k\Delta^k u_0(x,y,z).
\end{eqnarray*}
That ends the proof.
 \end{proof}
 
\section{Applications}
\subsection{Example 1}
We consider the RLW equation   \begin{eqnarray}
\dfrac{\partial u}{\partial t}+\dfrac{1}{2} \dfrac{\partial u^2}{\partial x}
=\dfrac{\partial^3 u}{\partial x^2\partial t},\,\,\,-\infty < x < +\infty,\,\, t> 0,
\end{eqnarray}
subjected to the initial condition
\begin{eqnarray}  u(x,t)|_{t=0}= x.  \end{eqnarray}
Using the recursive formula (\ref{3}), we have 
\begin{eqnarray} w_0^1&=&x  \\
 i \dfrac{\partial^2}{\partial x^2}w_{n+1}^1- iw_{n+1}^1&=&
\dfrac{1}{2}\dfrac{\partial }{\partial x}\sum_{k= 0}^n C_n^kw_k^1 w_{n-k}^1.
\end{eqnarray}
Finding the particular solutions of the previous recursive 
relation in the form $w_n^1=a_n x$, we obtain  
$w_1^1= ix,\,\,\, w_2^1= 2!i^2x,\,\,\,w_3^1= 3!i^3x,\,\,\,w_4^1
= 4!i^4x,\,\,\,w_5^1= 5!i^5x$. By induction we obtain $w_n^1= n!i^nx$.\\
By virtue of the expansion of $u(x,t),$ we get
\begin{eqnarray}
u(x,t)= \sum_{n\geq 0}\dfrac{(it)^n}{n!}n!i^nx= \sum_{n\geq 0}
\dfrac{(i^2t)^n}{n!}x=\sum_{n\geq 0}(-t)^nx=\dfrac{x}{1+t}.\label{10}
\end{eqnarray} 
We can easily see that (\ref{10}) is the exact solution of the initial value problem.

\subsection{Example 2}
We consider   \begin{eqnarray}
\dfrac{\partial u}{\partial t}+ \dfrac{\partial u}{\partial x}&=& 
2\dfrac{\partial^3 u}{\partial x^2\partial t},\,\,\,-\infty < x <+\infty,\,\, t> 0 \\
 u(x,t)|_{t=0}&=& e^{-x}. 
\end{eqnarray}
Using the recursive formula (\ref{3}), we obtain 
\begin{eqnarray}
 w_0^1&=&e^{-x} \\
2i \dfrac{\partial^2}{\partial x^2}w_{n+1}^1- iw_{n+1}^1&=&
\dfrac{\partial }{\partial x}w_n^1,\,\,\,n=0, 1, 2,3,\ldots.
\end{eqnarray}
Thus, determining the particular solutions of the recursive formula we have
$$w_1^1= ie^{-x};\;\; w_2^1= i^2e^{-x};\;\;w_3^1= i^3e^{-x};\;\;w_4^1
= i^4e^{-x}.$$ By induction we obtain $w_n^1= i^ne^{-x}.$
By means of the expansion of $u(x,t),$ we can write
\begin{eqnarray}
u(x,t)= \sum_{n\geq 0}\dfrac{(it)^n}{n!}i^ne^{-x}= \sum_{n\geq 0}\dfrac{(i^2t)^n}{n!}e^{-x}=\sum_{n\geq 0}
\dfrac{(-t)^n}{n!}e^{-x}=e^{-t-x}
\end{eqnarray}
as the exact solution of the initial value problem.
 \subsection{Example 3}
Find the solution of the initial value problem
  \begin{eqnarray}
\dfrac{\partial u}{\partial t}+ 2\dfrac{\partial^4 u}{\partial x^4}&=& 
\dfrac{\partial^3 u}{\partial x^2\partial t} \\
 u(x,t)|_{t=0}&=& \sin x. \end{eqnarray}
By virtue of the recursive formula (\ref{3}), we get 
\begin{eqnarray}  w_0^1&=& \sin x  \\
i \dfrac{\partial^2}{\partial x^2}w_{n+1}^1- iw_{n+1}^1&=&
2\dfrac{\partial^4 }{\partial x^4}w_n^1. 
\end{eqnarray}
Thus, $w_1^1= i\sin x;\;\;\;w_2^1= i^2\sin x;\;\; w_3^1= i^3\sin x;\;\; w_4^1
= i^4\sin x;\;\;w_5^1= i^5\sin x$. By induction,  we obtain $w_n^1= i^n\sin x$ and
\begin{eqnarray} u(x,t)= \sum_{n\geq 0}\dfrac{(it)^n}{n!}i^n\sin x= 
\sum_{n\geq 0}\dfrac{(i^2t)^n}{n!}\sin x=\sum_{n\geq 0}
\dfrac{(-t)^n}{n!}\sin x=e^{-t}\sin x \end{eqnarray}
 as the exact solution of the initial value problem.
 \subsection{Example 4}
Consider a homogeneous ball with constant radius $R$.
The temperature $T$ of this ball depending on the variable radius $r$ 
and the time $t$ is governed by the heat conduction equation \cite{Vladimirov}
\begin{eqnarray}
\frac{\partial T}{\partial t}= a^2 \Big(\frac{\partial^2 T}{\partial r^2}
+\frac{2}{r}\frac{\partial T}{\partial r}\Big) \label{11} \end{eqnarray}
subject to the initial condition
\begin{eqnarray} T(r,t)|_{t=0}= T_0(r) \label{12} \end{eqnarray}
and to boundary conditions
\begin{eqnarray} \Big(\frac{\partial T}{\partial r}(r,t) 
+h T(r,t)\Big)|_{r=R}= 0. \label{13} \end{eqnarray}
In order to solve the latter problem, we undertake the change of variables $V=rT$ and thus,
\begin{eqnarray}
\frac{\partial V}{\partial t}&=& a^2\frac{\partial^2 V}{\partial r^2}\label{14}\\
V(r,t)|_{t=0}&=& rT_0(r)\label{15}\\
V(r,t)|_{r=0}&=&0;\quad \Big[\frac{\partial V}{\partial r}(r,t) 
+\big(h-\frac{1}{R}\big)V(r,t)\Big]|_{r=R}= 0.\label{16}
\end{eqnarray}
By virtue of the recursive relation (\ref{9}), we have
\begin{eqnarray}  w_k= (-ia^2)^k\dfrac{d^{2k}}{dr^{2k}}[rT_0(r)]
\end{eqnarray} 
  and the solution can be expressed in the form
\begin{eqnarray}
V(r, t)= \sum_{n\geq 0}\dfrac{(it)^n}{n!}(-ia^2)^n\dfrac{d^{2n}}{dr^{2n}}[rT_0(r)]
= \sum_{n\geq 0}\dfrac{(a^2t)^n}{n!}\dfrac{d^{2n}}{dr^{2n}}[rT_0(r)]. \end{eqnarray}
Going back to the initial variable, we get
\begin{eqnarray}
T(r, t)= \dfrac{1}{r}\sum_{n\geq 0}\dfrac{(a^2t)^n}{n!}\dfrac{d^{2n}}{dr^{2n}}[rT_0(r)].
\end{eqnarray}
\subsection{Example 5}
We consider the time-dependent Navier-Stokes problem 
in the primitive variables \cite{Kharrat}\cite{Eglit}
\begin{eqnarray}
\frac{\partial u}{\partial t} - \nu \Delta u + \nabla p 
&=& f \quad \rm{in}\quad \Omega \times [0,T]\label{17}\\
 \nabla.u &=& 0\quad \rm{in}\quad \Omega \times [0,T]\label{18}\\
 u &=& 0\quad \rm{on}\quad \Gamma \times ]0,T[\label{19}\\
 u_{|t=0}&=& u_0 \quad \rm{in}\quad \Omega, \label{20}
\end{eqnarray}
where $\Omega $ is the bounded connected domain of $\mathbb{R}^d$ $(d=2, 3)$ 
with the Lipschitz continuous boundary $\Gamma$; the unknowns
are the velocity $u = (u_1(x, y, z, t), u_2(x, y, z, t), u_3(x, y, z, t))$ 
and the pressure $p=p(x, y, z, t)$; the data stand for
$f=(f_1(x, y, z, t), f_2(x, y, z, t), f_3(x, y, z, t))$ which represents 
the force on the body and $u_0=u_0(x, y, z) $, the initial condition; $\nu $ 
is the kinematic viscosity assumed to be a positive constant.

Using Bagarello's approach, we provide analytical solutions to the linearized 
Navier-Stokes equations (\ref{17}) satisfying the initial and boundary 
conditions (\ref{18})-(\ref{20}) and the pressure.

We can state the main result as:
\begin{theorem}
The velocity $u$ and the pressure $p$ of the incompressible fluid flow 
described by the Navier-Stokes equations are defined by
 \begin{eqnarray}
  u(x,y,z,t) &=& -\nabla\times\Delta^{-1}\psi(x, y, z, t) + 
  \nabla\varphi(x, y, z, t)\quad \rm{in}\quad \Omega\times [0,T] \label{21}\\
 p(x, y, z, t) &=& p_0 +\frac{\partial \varphi}{\partial t}(x_0, y_0, z_0, t_0) 
 - \nabla.(\Delta^{-1}f)(x_0, y_0, z_0, t_0)\nonumber\\
 &&-\frac{\partial \varphi}{\partial t}(x, y, z, t) 
 + \nabla.(\Delta^{-1}f)(x, y, z, t),\label{22} \end{eqnarray}
where  \begin{eqnarray}
\psi(x, y, z, t)&=& \sum_{n\geq 0}\dfrac{(\nu t)^n}{n!}\Delta^n
\big[\nabla\times u_0(x, y, z)\big] + \phi(x, y, z, t) \label{23}\\
 \Delta^{-1}v(x,y,z,t)&=& \int_0^T\int_{\Omega}(4\pi \tau)^{-3/2}
 \exp\bigg\{\dfrac{-|\omega-\xi|^2}{\tau}\bigg\}v(\xi,t)d\xi d\tau, \label{24}\\
 \omega&=&(x, y, z),  \nonumber \end{eqnarray}
 $\phi$ a particular solution of 
 \begin{eqnarray} \frac{\partial \psi}{\partial t} - \nu \Delta \psi 
&=& \nabla\times f \quad \rm{in}\quad \Omega \times [0,T],
\end{eqnarray}
$\varphi$ is a given harmonic potential and its derivatives are decreasing 
functions when $\sqrt{x^2 +y^2 + z^2}$ tends to $+\infty$. 
\end{theorem}
\begin{proof}
Indeed, applying the operator $\nabla$ to equation (\ref{17}) and setting $ \psi=\nabla\times u,$ 
 we have
\begin{eqnarray}
\frac{\partial \psi}{\partial t} - \nu \Delta \psi &=& \nabla\times f \quad \rm{in}\quad \Omega \times [0,T]\label{25}\\
 \psi_{|t=0}&=&\nabla\times u_0 \quad \rm{in}\quad \Omega. \label{26}
\end{eqnarray}
Setting $\nabla\times f =0$, substituting $ \psi=\sum_{n\geq 0}\dfrac{(it)^n}{n!}[H,Y^0]_n $ in the equation (\ref{25}) and by virtue of the previous theorem, we have 
 \begin{eqnarray}
[H,Y^0]_n=(-i\nu)^n\Delta^n Y^0,\,\,\, \psi(x, y, z, t)= \sum_{n\geq 0}\dfrac{(\nu t)^n}{n!}\Delta^n Y^0,\,\,\,Y^0=\nabla\times u_0.
\end{eqnarray} 
Therefore, we can express the solution $\psi $ to the problem (\ref{25}) and (\ref{26}) by
\begin{eqnarray}
\psi(x, y, z, t)= \sum_{n\geq 0}\dfrac{(\nu t)^n}{n!}\Delta^n Y^0 + \phi(x, y, z, t),
\end{eqnarray} 
where $\phi$ is a particular solution of equation (\ref{25}).\\
Applying the operator $\nabla$ to equation $\psi=\nabla\times u$ and taking account of equation(\ref{18}), we obtain the equation with unknown $u$
\begin{eqnarray}
\Delta u= -\nabla\times\psi
\end{eqnarray}
 which allows to get the solution $u$ in the form \cite{Maslov}
\begin{eqnarray}
 u(x, y, z, t)= -\nabla\times\Delta^{-1}\psi(x, y, z, t) + \nabla\varphi(x, y, z, t),
\end{eqnarray}
where the inverse operator $\Delta^{-1}$ of $\Delta$ is given by
\begin{eqnarray}
 \Delta^{-1}v(x,y,z,t)= \int_0^T\int_{\Omega}(4\pi \tau)^{-3/2}\exp\bigg\{\dfrac{-|\omega-\xi|^2}{\tau}\bigg\}v(\xi,t)d\xi d\tau,\,\,\,\omega=(x, y, z),
\end{eqnarray}
 $\varphi$ a given harmonic potential of the fluid and its derivatives are decreasing when 
 $\sqrt{x^2+ y^2+z^2}$ approaches $+\infty$. Then we derive the pressure $p$ from equation(\ref{17}) as
\begin{eqnarray*}
\nabla p &=& f -\frac{\partial u}{\partial t} + \nu \Delta u\Longleftrightarrow\\
\nabla p &=& f-\frac{\partial}{\partial t}\big(-\nabla\times\Delta^{-1}\psi + \nabla\varphi\big)-\nu\nabla\times \psi\\
  &=& f -\frac{\partial }{\partial t}(\nabla\varphi)
 + \Delta^{-1}\Big(\nabla\times\dfrac{\partial\psi}{\partial t}\Big)-\nu\nabla\times\psi \\
&=& f -\frac{\partial }{\partial t}(\nabla\varphi)
+ \Delta^{-1}\Big\{\nabla\times\Big[\dfrac{\partial\psi}{\partial t}-\nu\Delta\psi\Big]\Big\}\\
&=& f -\frac{\partial }{\partial t}(\nabla\varphi)
+ \Delta^{-1}\Big[\nabla\times\nabla\times f\Big]\\
& =&\nabla\Big[-\frac{\partial\varphi }{\partial t} + \nabla.\big(\Delta^{-1}f\big)\Big]
\end{eqnarray*}
which reduces to
\begin{eqnarray}
  p = -\frac{\partial \varphi}{\partial t} + \nabla.(\Delta^{-1}f) + C,\quad C=\rm{const}.
\end{eqnarray}
For measured stagnation pressure $p(x, y, z, t)=p_0 $ at the point $(x_0, y_0, z_0, t_0),$  
  the pressure $p$ exerted on the fluid flow is defined by
\begin{eqnarray*}
p(x, y, z, t)&=& p_0 +\frac{\partial \varphi}{\partial t}(x_0, y_0, z_0, t_0)
 - \nabla.(\Delta^{-1}f)(x_0, y_0, z_0, t_0)\\
 &&-\frac{\partial \varphi}{\partial t}(x, y, z, t) + \nabla.(\Delta^{-1}f)(x, y, z, t).
\end{eqnarray*}
That ends the proof.
\end{proof}
Without loss of generality, we set  
\begin{eqnarray*}
\nabla\times u_0(x,y, z)&=& (cos y\cos z, \sin(x-y-z), e^{x+y+z}),\\
\nabla\times f(x, y, z,t)&=& (t\cos x, e^t, tz\sin x)   
\end{eqnarray*}
and thus obtain
\begin{eqnarray}
\psi(x,y,z,t)= \begin{pmatrix}
e^{-2\nu t}\cos y\cos z+\nu^{-2}(-1+\nu t+e^{-\nu t})\cos x\\
e^{-3\nu t}\sin(x-y-z)+e^t -1\\
e^{3\nu t+x+y+z}+ \nu^{-2}(-1+\nu t+e^{-\nu t})z\sin x
\end{pmatrix}
\end{eqnarray}
\begin{eqnarray}
\nabla\times \psi(x,y,z,t)=\begin{pmatrix}
e^{3\nu t+x+y+z}+e^{-3\nu t}\cos(x-y-z)\\
 - e^{3\nu t+x+y+z}-e^{-2\nu t}\cos y\sin z
 -\nu^{-2}(-1+\nu t+e^{-\nu t})z\cos x \\
  e^{-2\nu t}\sin y\cos z + e^{-3\nu t}\cos(x-y-z)
\end{pmatrix} 
\end{eqnarray}
\begin{eqnarray}
 \Delta^{-1}(\nabla\times\psi(x,y,z,t))= \int_0^T\int_{\Omega}(4\pi\tau)^{-3/2}\exp\bigg\{\dfrac{-|\omega-\xi|^2}{\tau}\bigg\}(\nabla\times\psi(\xi,t))d\xi d\tau.
\end{eqnarray}
Choosing the harmonic potential $\varphi$ in the form
\begin{eqnarray}
\varphi(x,y,z,t)=\dfrac{t}{\sqrt{x^2+y^2+z^2}},\,\,\,(x, y, z)\neq (0,0,0),
\end{eqnarray}
we can express the velocity $u$ and the pressure $p$ as 
\begin{eqnarray*}
u(x,y,z,t)&=&
-\int_0^T\int_{\Omega}(4\pi \tau)^{-3/2}\exp\bigg\{\dfrac{-|\omega-\xi|^2}{\tau}\bigg\}(\nabla\times\psi(\xi,t))d\xi d\tau\\
&&-t\bigg(x (x^2+y^2+z^2)^{-3/2}, y(x^2+y^2+z^2)^{-3/2},  z(x^2+y^2+z^2)^{-3/2}\bigg)\\
p(x, y, z, t)&=& p_0 +
\dfrac{1}{\sqrt{x_0^2 + y_0^2 + z_0^2}}- \dfrac{1}{\sqrt{x^2 + y^2 + z^2}}.
\end{eqnarray*}

\section{Concluding remarks}

The non-commutative strategy developed by Bagarello
 for the analysis of systems of ordinary differential 
equations (ODEs) has been extended to a class of 
partial differential equations. Systems of 
PDEs have been solved using an unbounded self-adjoint, densely defined, Hamiltonian operator and  a recursion relation which provides  a multiple 
commutator and  a power series solution.
Numerous examples like the  equations of diffusion, RLW,  heat conduction and Navier-Stokes equations have been given in this work proving the efficiency  of considered method also for the analysis of relevant PDEs widely used in mathematical physics. In that sense, this work can be considered as a good supplement  to Bagarello's approach.


\begin{thebibliography}{9}                                                                                                %


\bibitem {Agoshkov} Agoshkov, V. I., Dubovski, P. B., Shutyaev, V. P., \textit{Methods for solving 
mathematical physics problems.} Cambridge International Science Publishing. London. (2006)

\bibitem{Bagarello97} Bagarello, F., A Non-Commutative approach to ordinary differential equations. 
				Int. Jour. of Theoretical Physics, \textbf{43}(12), 2371-2394 (2004).
\bibitem{Chengri} Jin, C., Liu, M., A new modification of Adomian decomposition method for solving 
a kind of evolution equation. Applied Mathematics and Computation \textbf{169}(2), 953-962 (2005) .

\bibitem{Chesnokov} Chesnokov, A. A.,  Symmetries and exact solutions of the rotating shallow-water 
equations. European Journal of Applied Mathematics, \textbf{20}, 461-477 (2009). doi:10.1017/S0956792509990064. 

\bibitem{Eglit} Eglit, M. E.,  Probl\`emes de m\'ecanique des milieux continus, ed., Lycee Moscovite, 
Moscou, Tomes 1,2.(1996).
%
%

\bibitem{Kharrat} Kharrat M., Mghazli, Z.,  Adaptive Algorithm for Chorin's scheme. 
Application to the linearized Navier-Stokes equations. CARI 2008 Maroc, 323-330 (2008)

%

\bibitem {Lu} Lu, L. Doering, C. R., Busse, F. H.,  Bounds on convection driven by internal heating.
J. Math. Phys. \textbf{45} (7),  2968-2985 (2004). 
\bibitem {Maslov} Maslov, V. P., Asymptotic methods for solving pseudodifferential 
equations Nauk Publ. Moscow (1987). (in Russian)



\bibitem{Vladimirov} Vladimirov V. S.,  A collection of problems on the equations of 
mathematical Physics (in English). Springer, (1986). 

\bibitem{Zadrzynska} Zadrzynska, E., Zajackowski, W. M.,  Global Regular Solutions with Large 
Swirl to the Navier-Stokes Equations in a Cylinder. Journal of Mathematical Fluid Mechanics \textbf{11}(1), 126-169 (2009).
\end{thebibliography}
\end{document}